\documentclass[showpacs,pre,preprint,preprintnumbers,amsmath,amssymb,letter]{revtex4-2}

\usepackage[T1]{fontenc}
\usepackage[utf8]{inputenc}
\usepackage{units}
\usepackage{amsmath}
\usepackage{amsthm}
\usepackage{graphicx}

\providecommand{\tabularnewline}{\\}

\usepackage{dcolumn}% Align table columns on decimal point
\usepackage{bm}% bold math

\usepackage{float}
\usepackage{color}
\usepackage{array}
\usepackage{multirow}
\usepackage[title]{appendix}

\usepackage{nicefrac}
\usepackage{ulem}
\usepackage{currvita}
\usepackage{amsthm}
\usepackage{inputenc}

\newtheorem*{theorem}{}
\newtheorem*{corollary}{}

\begin{document}
\preprint{APS/PRE}
\title{Energy conversion theorems for some linear steady--states}

\author{L. A. Arias--Hernandez}
\email{larias@ipn.mx}
\affiliation{Departamento de F\'{\i}sica, Escuela Superior de F\'{\i}sica y Matem\'aticas,
Instituto Polit\'ecnico Nacional, U. P. Zacatenco, Edif. 9, 2o Piso,
Ciudad de México, 07738, México.}

\author{G. Valencia--Ortega}
\email{gvalenciao@ipn.mx}
\affiliation{Departamento de F\'{\i}sica, Escuela Superior de F\'{\i}sica y Matem\'aticas,
Instituto Polit\'ecnico Nacional, U. P. Zacatenco, Edif. 9, 2o Piso,
Ciudad de México, 07738, México.}

\author{F. Angulo--Brown}
\email{fangulob@ipn.mx}
\affiliation{Departamento de F\'{\i}sica, Escuela Superior de F\'{\i}sica y Matem\'aticas,
Instituto Polit\'ecnico Nacional, U. P. Zacatenco, Edif. 9, 2o Piso,
Ciudad de México, 07738, México.}

\author{C. R. Martinez--Garcia}
\email{cemartinezg@ipn.mx}
\affiliation{Departamento de Ciencias B\'asicas, Escuela Superior de C\'omputo, Instituto
Polit\'ecnico Nacional, Av. Miguel Bernard, Esq. Av. Miguel Oth\'on de
Mendizabal, Colonia Lindavista, Ciudad de M\'exico 07738, M\'exico.}

%\date{}

\begin{abstract}
One of the main issues that real energy converters present, when they
produce effective work, is the inevitable entropy production. Within
the context of Non-equilibrium Thermodynamics, entropy production
tends to energetically degrade man-made or living systems. On the
other hand, it is also not useful to think about designing an energy
converter that works in the so-called minimum entropy production regime
since the effective power output and efficiency are zero. In this
manuscript, we establish some \textit{Energy Conversion Theorems}
similar to Prigogine's one with constrained forces, their purpose
is to reveal trade-offs between design and the so-called operation
modes for $\left(2\times2\right)$--linear isothermal energy converters.
The objective functions that give rise to those thermodynamic constraints
show stability. A two--meshes electric circuit was built as an example
to demonstrate the Theorems' validity. Likewise, we reveal a type
of energetic hierarchy for power output, efficiency and dissipation
function when the circuit is tuned to any of the operating regimes
studied here: maximum power output ($MPO$), maximum efficient power
($MP\eta$), maximum omega function ($M\Omega$), maximum ecological
function ($MEF$), maximum efficiency ($M\eta$) and minimum dissipation
function ($mdf$). 
\end{abstract}


\maketitle

\section{\label{int}Introduction}

Since Prigogine formulated his principle of minimum entropy production
in 1947 \cite{Prigogine47}, also known as Prigogine's theorem, it
has been subject to several controversies \cite{Tykodi67,Jaynes80,Mamedov03,Martyushev07,Bertola08,Martyushev13}.
This theorem states that in the linear regime, where the Onsager reciprocal
relations are valid \cite{Onsager31I}, all steady states in which
unconstrained thermodynamic flows vanish are characterized by the
following extremum principle, \textit{``...In the linear regime,
the total entropy production in a system subject to flow of energy
and matter, $\nicefrac{d_{i}S}{dt}=\int\sigma dV$, reaches a minimum
value at the non-equilibrium stationary state...''} \cite{KondePrigo98}.
Despite the criticism received, the Prigogine's theorem has prevailed
mainly due to its experimental verifications; for example, in case
of heat conduction in metallic rods \cite{Rafols92,Danielewicz00},
as well as computer simulations of the same system \cite{Lurie80}.

In this last article, to demonstrate the validity of Prigogine's theorem,
Luri\'e and Wagensberg took as the only fixed force, the following temperature
gradient $F_{0}=T_{h}^{-1}-T_{0}^{-1}$, i.e, it considers extreme
thermal reservoir temperatures of the rod. The rest of the $(n-1)$
slices temperatures of their discrete model are used to construct
the free forces $F_{k}$, for $k=1,2,...,(n-1)$. Under this assumption,
they arrive to the minimum entropy production regime by following
all the steps of the Prigogine procedure. It is important to note
that in \cite{Lurie80}, the so-called phenomenological coefficients
$L$ were used to represent the Fourier Law in the unidimensional
form $J(x,t)=LX(x,t)$, where $L=kT^{2}$. That is, the thermal conductivity
depends on $T^{-2}$ and as was asserted by Jaynes \cite{Jaynes80},
there is no known substance which obeys this relation. This obstacle
was surrounded by Luri\'e and Wagensberg by considering that for small
enough temperature gradients the effects of taking $L=kT^{2}$ are
not important. This claim works reasonably. However, in \cite{KondePrigo98}
by means of the Euler-Lagrange formalism, the authors arrive at the
function $T(x)$, which minimizes the entropy production and is also
linear with respect to $x$-variable, by taking $L=kT^{2}\approx kT_{av}^{2}$,
where $T_{av}$ is the average temperature of the rod. If the temperature
gradient is small enough in the rod one will have that $T(x)=T_{av}\left[1+\varepsilon(x)\right]$
with $|\varepsilon(x)|\ll1$ \cite{Mamedov03}. For example, in \cite{Danielewicz00}
the corresponding experiment was performed for $\Delta T=341.5\text{K}-250\text{K}=51.5\text{K}$
and for each case $|\varepsilon(x)|\sim10^{-4}$ and the above-mentioned
approximation is hold. Another strong support in favor of Prigogine's
theorem in \cite{Klein54}, is that Klein and Meijer proved it by
using statistical mechanics methods. They assumed a process consisting
of mass and energy fluxes through a narrow tube that in turn connect
two containers of an ideal gas. Such as occurs when a gas is enclosed
by rigid adiabatic walls, which at the same time is connected by means
of a diathermic piston with a reservoir of both temperature and pressure,
this type of systems reach the final equilibrium state with their
respective reservoirs without performing work. Something analogous
happens in processes towards the steady state described in the previous
heat conduction examples. In the case of the aforementioned gas interacting
with the temperature and pressure reservoir, and by using the concept
of thermodynamic availability and coupling with a second system give
rise to the maximum useful work theorem \cite{Callen85,Mandl91}.
Similarly, within the context of Linear Irreversible Thermodynamics
(LIT), a large number of steady states arise from the coupling of
spontaneous processes and non-spontaneous ones. These couplings can
exist in both living and man-made systems. Remarkably, this coupling
concept can lead us to describe energy conversion processes in systems
with constrained forces solely \cite{Stucki80,CaplanEssig83,AriPaeAng08}.
Caplan and Essig \cite{CaplanEssig83} developed a theory based on
LIT for the study of linear biologic energy converters, that work
in steady states. These authors introduced concepts such as power
output and efficiency with the purpose to optimize the energy conversion
process. Likewise, they took the usual notion of entropy production.
Later, Stucki \cite{Stucki80} used those ideas in the analysis of
oxidative phosphorylation to establish other working regimes different
than the minimum entropy production. In addition, Arias-Hernandez
et al \cite{AriPaeAng08} used concepts from Finite-Time Thermodynamics
to analyze the above-mentioned energy conversion process. It is clear
that in heat conduction experiments subjected to small temperature
gradients, the systems evolves towards a steady final state of minimal
entropy production. Yet, as several authors have been shown \cite{Stucki80,CaplanEssig83,AriPaeAng08,AngSantCall95,SantAriAng97},
in the case of two or more coupled processes, there are other steady
states where certain quantity of interest to be optimized may come
from natural or artificial needs.

The foregoing can translate into different proposals of trade-offs
through characteristic functions, which are described in Section 2
of this article, i.e, thermodynamic mechanisms can be assumed whose
goals are to maximize some energetic objective functions. Thus, at
least five similar theorems to the Prigogine's one can be stated (Section
3), whose purpose is to show physical restrictions for $\left(2\times2\right)$--linear
isothermal energy converters to operate in some optimal and stable
regimes. In Section 4, we design a two-mesh electrical circuit to
exhibit nonzero energy conversion when the coupling of electrical
currents meet the physical constraints imposed by these ``Energy
Conversion Theorems'' and experimental verification thereof is presented.
Finally, our conclusions are exposed in Section 5.

\section{\label{steady}Steady states without minimum entropy production}

\begin{figure}
\begin{centering}
\includegraphics[scale=0.75]{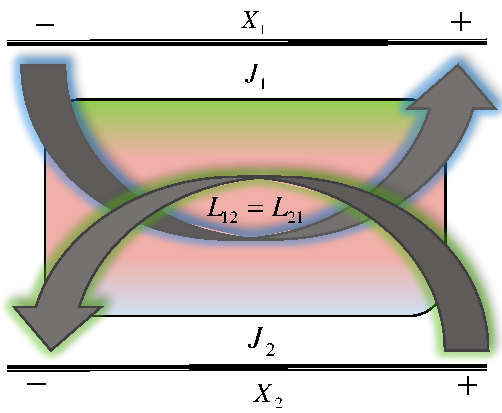} 
\par\end{centering}
\caption{\label{fig:energconv22}Sketch of a simple isothermal energy converter
(2 coupled fluxes promoted by 2 conjugated forces), where $X_{1}<0$
is the force associated with the non-spontaneous flux, while $X_{2}>0$
is the another one related to the spontaneous flux.}
\end{figure}

Several phenomena in nature that are related to the transport of mass,
charge and energy have been described with a good approximation in
the linear regime \cite{Onsager31I,Onsager31II,Callen1,OdumPink55}.
In general, the characteristic functions, that have been defined to
study non-equilibrium processes in a variety of man-made and living
systems, can be expressed in terms of sums of conjugate fluxes ($J_{i}$'s)
and forces ($X_{i}$'s) \cite{Stucki80,CaplanEssig83,AriPaeAng08,AngSantCall95,SantAriAng97,Fronczak07,Volkenstein83}.
Every $J_{i}$'s are the so-called thermodynamic fluxes, such as the
strain rate, reaction speed, electrical current, rate of muscle contraction,
etc. The $X_{i}$'s are defined as the thermodynamic forces, such
as the stress tensor, reaction affinities, electrochemical potential,
muscle tension, etc.

Most of the previous simple linear relations are well-known, for example,
Ohm's Law for electrical current, Fick's Law for diffusion, Fourier
Law for heat flow, etc. The main feature of these uncoupled processes
promoted by unconstrained forces, is their indisputable contribution
to the increase of system’s entropy and a decrease in its free energy,
i.e, the unconstrained and spontaneous fluxes emulate simple diffusive
and non-interacting processes \cite{Sekimoto,SasaTasa06}. However,
there are other cases, characterized by multiple interacting fluxes
subjected to constrained forces. This type of phenomena involve energy
conversion processes, which can be classified of two sets of fluxes:
the one associated with an entropy increase (spontaneous) and the
other one with a decrease of entropy (non-spontaneous) \cite{AriPaeAng08,ValAri17,GonAri19,ValAri20}.
To illustrate it, we can take as an example the transport of an ion
across a cell membrane, which may be influenced not only by its electrochemical
gradient, but also by the influence of other gradient, such as a external
pressure.

In general to characterize a type of processes that occur in open
thermodynamic systems is through the so-called steady states. These
steady states are typical of systems whose processes are kept constant
on time, so the entropy created by the steady flow is more relevant
than the entropy transferred to the surroundings \cite{CaplanEssig83,SantAriAng97,DeGroot62,Prigogine67},

\begin{equation}
\frac{dS_{T}}{dt}=\frac{dS_{int}}{dt}+\frac{dS_{ext}}{dt}>0,\label{eq:entrprod}
\end{equation}
where $\dot{S}_{int}\equiv\sigma=\hat{J}:\hat{X}$ is usually called
the entropy production while $\dot{S}_{ext}=0$, since the entropy
flux from the surroundings is equal to the entropy flux toward the
system. That is, internal irreversibilities are responsible for the
total entropy increments of the thermodynamic universe.

In the case of $(2\times2)$--isothermal linear energy converter
models, the entropy production is also a positive semidefinite quadratic
form \cite{Onsager31I,KondePrigo98,CaplanEssig83,AriPaeAng08,SantAriAng97,Onsager31II,ValAri17,ColinGolds03}:

\begin{eqnarray}
\sigma & = & J_{1}X_{1}+J_{2}X_{2}\nonumber \\
 & = & \left[X_{1},X_{2}\right]\left[\begin{array}{cc}
L_{11} & q\sqrt{L_{11}L_{22}}\\
q\sqrt{L_{11}L_{22}} & L_{22}
\end{array}\right]\left[\begin{array}{c}
X_{1}\\
X_{2}
\end{array}\right]>0,\label{eq:pentrmatr}
\end{eqnarray}
where the coefficient matrix $\hat{L}$ is symmetric. From this quadratic
form, it can be defined the ``degree of coupling'': $q=\nicefrac{L_{12}}{\sqrt{L_{11}L_{22}}}$,
which fulfills $q^{2}\in\left[0,1\right]$, with $L_{11}>0$ and $L_{22}>0$,
this coefficient measures the energy conversion quality \cite{Stucki80,CaplanEssig83,AriPaeAng08}.

Since the purpose of those models is describing energy conversion
phenomena, three energetic functions (process variables) with extreme
conditions can be defined, the dissipation function ($\Phi\equiv T\sigma$),
power output ($P_{O}\equiv-TJ_{1}X_{1}$) and efficiency ($\eta\equiv\nicefrac{P_{O}}{P_{I}}=-\nicefrac{TJ_{1}X_{1}}{TJ_{2}X_{2}}$),

\begin{subequations} \label{eq:potefic} 
\begin{eqnarray}
\Phi & = & \left(x^{2}+2qx+1\right)TL_{22}X_{2}^{2}\\
P_{O} & = & -x\left(x+q\right)TL_{22}X_{2}^{2}\\
\eta & = & -\frac{\left(x+q\right)x}{qx+1},
\end{eqnarray}
\end{subequations} here we introduce a performance parameter $x=\sqrt{\nicefrac{L_{11}}{L_{22}}}\nicefrac{X_{1}}{X_{2}}$,
called the force ratio \cite{Tribus61}. It measures the cross effect
between two potentials \cite{CaplanEssig83,AriPaeAng08}.

\begin{figure}
\begin{centering}
\includegraphics[scale=0.51]{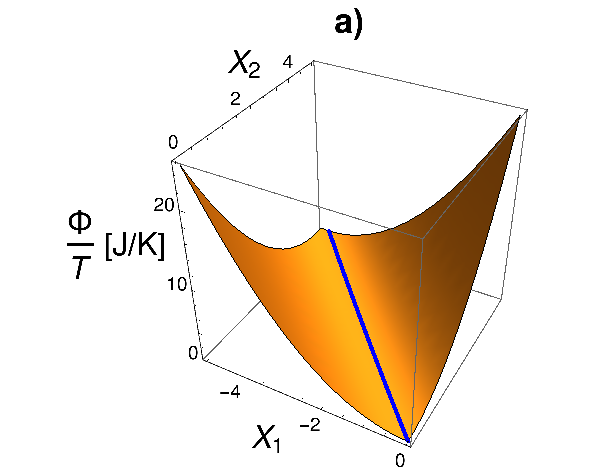} %\linebreak
\includegraphics[scale=0.55]{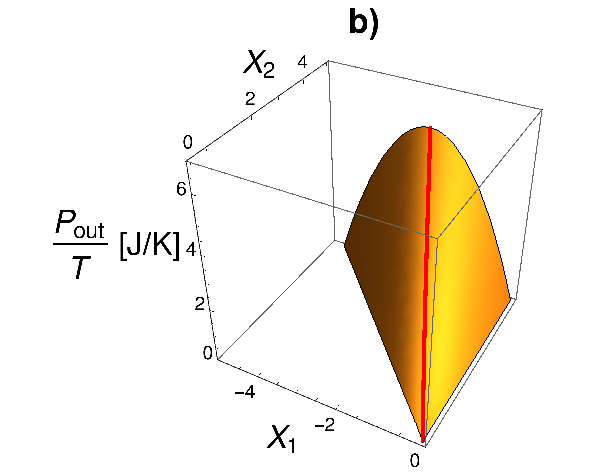} %\linebreak
\includegraphics[scale=0.55]{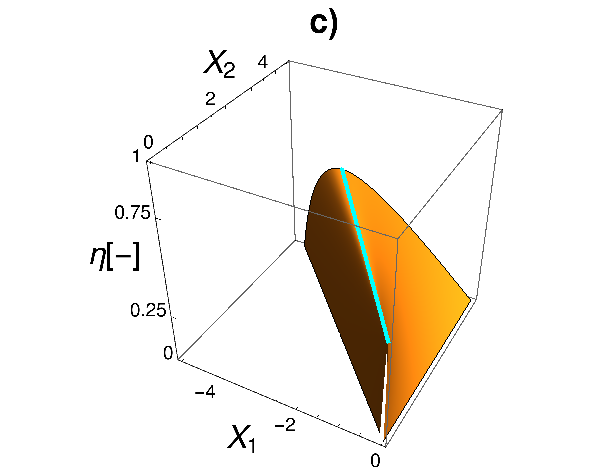} 
\par\end{centering}
\caption{\label{fig:dispotefic}Energetic functions versus $X_{1}$ (driven
force) and $X_{2}$ (driver force) for an isothermal linear energy
converter. Normalized dissipation function (a), normalized power output
(b) and efficiency (c), all of them are plotted for a fixed value
of $q=0.97$. They reach their extreme values at $X_{1}^{mdf}=-\sqrt{\nicefrac{L_{22}}{L_{11}}}\,X_{2}\,q$
(blue line), $X_{1}^{MPO}=-\sqrt{\nicefrac{L_{22}}{L_{11}}}\,X_{2}\,\nicefrac{q}{2}$
(red line) and $X_{1}^{M\eta}=-\sqrt{\nicefrac{L_{22}}{L_{11}}}\,X_{2}\,\nicefrac{q}{\left(1+\sqrt{1-q^{2}}\right)}$
(cyan line), respectively.}
\end{figure}

Under this linear energy converters scheme (see Fig. \ref{fig:energconv22})
\cite{CaplanEssig83,AriPaeAng08,OdumPink55}, $X_{1}$ can be defined
as the driven force while $X_{2}$ as the driver force, that is, $X_{1}=X_{1}(q,X_{2})$.
Fig \ref{fig:dispotefic} shows the extremes in the three functions
($\Phi$, $P_{O}$ and $\eta$) at different values of $\left(X_{1},X_{2}\right)$,
and at the same time, they are associated to three different operation
modes, the minimum dissipation function ($mdf$), maximum power output
($MPO$) and maximum efficiency ($M\eta$). The optimal values that
$x$ adopts for each of the process variables in Eq. \ref{eq:potefic}
are: \begin{subequations} \label{eq:opXvp} 
\begin{eqnarray}
x^{mdf} & = & -q\\
x^{MPO} & = & -\frac{q}{2},\\
x^{M\eta} & = & -\frac{q}{1+\sqrt{1-q^{2}}};
\end{eqnarray}
\end{subequations} consequently, other objective functions express
different trade--offs between the energetic functions $\Phi$, $P_{O}$
and $\eta$ can be defined, therefore they also have extreme values
(see Fig. \ref{fig:funobj}). The objective functions that we will
study to characterize their optimal operation modes are \cite{Stucki80,SantAriAng97,Calvoetal01}:
\begin{subequations} \label{eq:compfunc} 
\begin{eqnarray}
E_{F} & = & -\left(2x^{2}+3xq+1\right)TL_{22}X_{2}^{2}\\
\Omega & = & \left[\eta_{M}\left(qx+1\right)-2x\left(x+q\right)\right]TL_{22}X_{2}^{2}\\
P_{\eta} & = & \frac{\left[\left(x+q\right)x\right]^{2}}{qx+1}TL_{22}X_{2}^{2},
\end{eqnarray}
\end{subequations} where $E_{F}=P_{O}-\Phi$ is the ecological function
\cite{Angulo91}, $\Omega=\left(2-\nicefrac{\eta_{M}}{\eta}\right)P_{O}$
is the so--called omega \cite{Calvoetal01}, with $\eta_{M}$ the
value of efficiency under the conditions that maximize it, i.e., Eq.
\ref{eq:potefic}c evaluated at Eq. \ref{eq:opXvp}c, and $P\eta=\eta P_{O}$
is the efficient power \cite{Stucki80,Yilmaz06}. Each of the previous
objective functions {[}Eqs. (\ref{eq:potefic}) and (\ref{eq:compfunc}){]}
give us optimal performance modes, such as maximum ecological function
($MEF$), maximum $\Omega$--function ($M\Omega$) and maximum efficient
power ($MP\eta$), for a linear energy converter, and their optima
$x$--values are, \begin{subequations} \label{eq:opXvp2} 
\begin{eqnarray}
x^{MEF} & = & -\frac{3q}{4}\\
x^{M\Omega} & = & -\frac{q\left(4-q^{2}+4\sqrt{1-q^{2}}\right)}{4\left(1+\sqrt{1-q^{2}}\right)^{2}}\\
x^{MP\eta} & = & -\frac{4+q^{2}-\sqrt{16-16q^{2}+q^{4}}}{6q}.
\end{eqnarray}
\end{subequations}

The above--mentioned physically accessible characteristic points
can be viewed in a $P_{O}$ vs. $\eta$ plane (see Fig. 4 of \cite{AriPaeAng08}),
in a similar way to the behavior of heat engines that operate in cycles
between two thermal reservoirs \cite{AriPaeAng08,ChenSun04}.

\subsection{Optimal performance modes within $\Phi$ vs. $\eta$ and $P_{O}$
vs. $\eta$ spaces}

\begin{figure}
\begin{centering}
\includegraphics[scale=0.45]{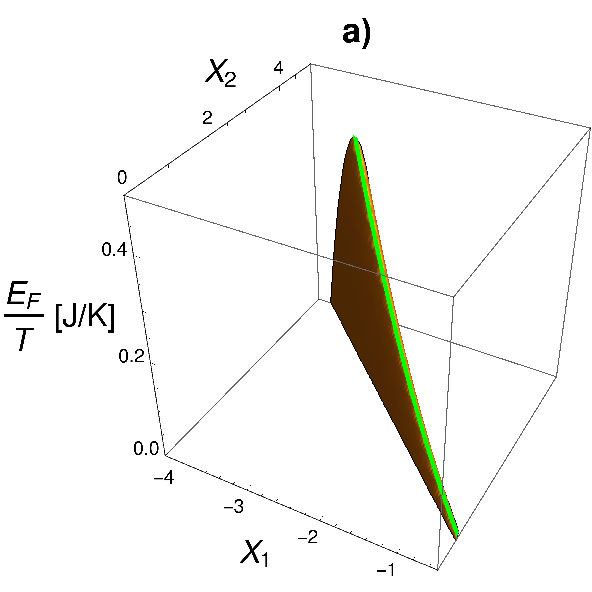} %\linebreak
\includegraphics[scale=0.58]{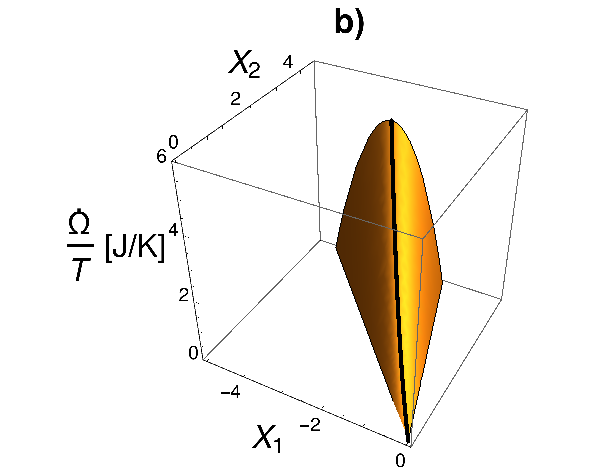} %\linebreak
\includegraphics[scale=0.58]{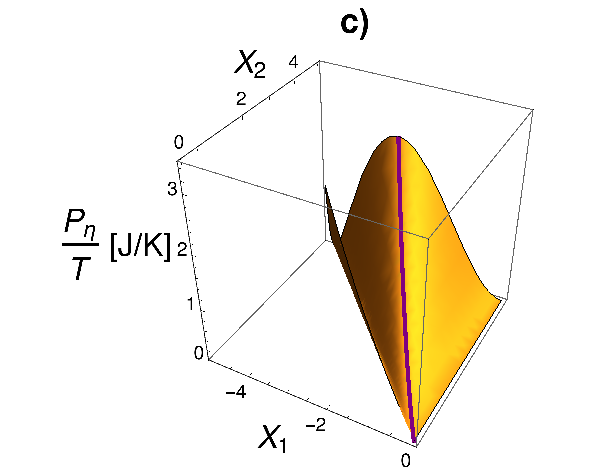} 
\par\end{centering}
\caption{\label{fig:funobj}Objective functions versus $X_{1}$ (driven force)
and $X_{2}$ (driver force) for an isothermal linear energy converter.
Normalized ecological function (a), normalized omega function (b)
and normalized efficient power (c), all of them are plotted for a
fixed value of $q=0.97$. They reach their extreme values at $X_{1}^{MEF}$
(green line), $X_{1}^{M\Omega}$ (black line) and $X_{1}^{MP\eta}$
(purple line), respectively.}
\end{figure}

In describing the performance of a linear energy converter, the process
variables can define a configuration space to display all physically
possible realizations. Then, $\Phi=\Phi(\eta)$ and $P_{O}=P_{O}(\eta)$,
this parametric way of representing the dissipation and the power
output of a linear energy converter, for $q\in\left[q_{min},1\right]$,
allows us to note that the dissipation function has an inverse decreasing
monotonous behavior, that is, when $q$ decreases the limit value
of $\Phi$ increases while $\eta$ diminishes \cite{AriPaeAng08}.
From the parametric graphs depicted in \cite{AriPaeAng08,ValAri20,Apertet13,Garcia08},
we can observe that when the quality of the coupling between spontaneous
and non-spontaneous fluxes is low, entropy production increases. On
the other side, the well-known loop--shaped curves for heat engines
\cite{HulGor92} arise in LIT energy converters when $P_{O}=P_{O}\left(\eta,q\right)$.
These loops intrinsically show the existence of extra--thermodynamic
conditions that dynamically constrain the processes of a linear energy
converter. Each loop has two optimal points of interest \cite{AriPaeAng08,ValAri20},
the one that corresponds to maxima $P_{O}$ points and the other one
to maximum-$\eta$ points, between these two points the other performance
regimes can be achieved \cite{LevValBar20}. This zone reveals energy
conversion with high power output, high efficiency, and low dissipated
energy.

When the operation modes associated with the three trade-off functions,
are evaluated in the so-called process variables ($\Phi$, $P_{O}$,
$\eta$) the following hierarchies are established (see Appendix A),
\begin{subequations} \label{eq:hierDPE} 
\begin{eqnarray}
\hspace*{-0.75cm}\Phi_{mdf}<\Phi_{M\eta}<\Phi_{MEF} & < & \Phi_{M\Omega}<\Phi_{MP\eta}<\Phi_{MPO}\\
\hspace*{-0.75cm}P_{O}^{mdf}<P_{O}^{M\eta}<P_{O}^{MEF} & < & P_{O}^{M\Omega}<P_{O}^{MP\eta}<P_{O}^{MPO}\\
\hspace*{-0.75cm}\eta_{mdf}<\eta_{MPO}<\eta_{MP\eta} & < & \eta_{M\Omega}<\eta_{MEF}<\eta_{M\eta}
\end{eqnarray}
\end{subequations}

The infinite optimal performance regimes located between $\left[P_{O}^{MPO},\,P_{O}^{M\eta}\right]$
can be achieved when the flux associated with the driven force is
tuned by means of the coupling coefficient with the flux associated
with the driver force. That is, operation modes are linked with objective
functions, which in principle can be built through the process variables.

\section{Energy conversion theorems for some linear steady--states}

In general, the validity limits for the hypotheses of LIT can be experimentally
verified, using several non-equilibrium systems characterized by continuous
variables \cite{KondePrigo98,Danielewicz00,Miller60,Hanley69}. In
particular, the Onsager relations have been shown to reflect the experimental
results for small thermodynamic gradients, that is, the postulates
of LIT are valid in situations close to some steady state \cite{Onsager31I,KondePrigo98,Onsager31II,DeGroot62,Prigogine67}.

In systems whose purely spontaneous processes are promoted by constrained
thermodynamic forces, the entropy production will always adapt to
a condition in which the characteristic steady state causes the same
systems to dissipate the least energy possible to the surroundings
(minimum entropy production's principle) \cite{KondePrigo98,Volkenstein83,DeGroot62,Prigogine67}.
On the other hand, energy converters open up a range of physically
accessible steady states, which represent thermodynamic constraints
under boundary conditions that correspond to optimal operating criteria
\cite{Stucki80,AriPaeAng08,AngSantCall95,SantAriAng97,ValAri17,GonAri19}.
These processes are characterized by a set of fixed forces ($X_{i}>0$),
associated with spontaneous fluxes, and another set of constrained
forces ($X_{j}<0$), subjected to an external condition and associated
with non-spontaneous fluxes.

For the case of energy converters with two constrained forces and
two coupled fluxes, a ``constrained minimum entropy production theorem''
can be claimed in terms of the dissipation function ($\Phi$). This
new proposal can be also stated as the ``minimum dissipation function
theorem'' \textbf{(mdf}--\textbf{THEOREM}), as follows \cite{ColinGolds03,Garcia08}:

\begin{theorem} \textbf{\textsl{\large{}mdf--Theorem}} When a non
equilibrium steady--state system is characterized by two generalized
forces $X_{1}$ (the one associated with driven processes) and $X_{2}$
(another one associated with driver processes), it is driven to a
steady state when the driver force is fixed. Then, under the condition
of \textbf{minimum dissipation function}, the $J_{1}$ flux vanishes.
\end{theorem}
\begin{proof}
Let us take the mathematical expression for dissipation function {[}Eq.
(\ref{eq:pentrmatr}){]}. By calculating the derivative of $\Phi$
with respect to $X_{1}$ assuming $X_{2}$ fixed, we obtain: 
\begin{equation}
\left(\frac{\partial\Phi}{\partial X_{1}}\right)_{X_{2}}=T\left[J_{1}+L_{11}X_{1}+q\left(\sqrt{L_{11}L_{22}}\right)X_{2}\right].\label{partialphi}
\end{equation}

By hypothesis $\left(\frac{\partial\Phi}{\partial X_{1}}\right)_{X_{2}}=0$;
therefore: 
\begin{equation}
J_{1}=0.\label{eq:condforflux1mfd}
\end{equation}
\end{proof}
In order to establish a trade--off between design and operation mode
of a linear energy converter, we will set out the following corollary
that results from canceling the driven process.

\begin{corollary} \textbf{\textsl{\large{}mdf--COROLLARY}} If the
degree of coupling between the processes of a non equilibrium steady--state
system is measured by the coefficient $q$ and it is also operating
under the \textbf{mdf}--regime, the cross effect between both generalized
forces, denoted by the parameter $x$, is given by $x=-q$. \end{corollary}
\begin{proof}
From the \textbf{$mdf$--THEOREM}, the constraint $J_{1}=0$ lead
us to write the force $X_{1}$ as: 
\begin{equation}
X_{1}=-q\left(\sqrt{\frac{L_{22}}{L_{11}}}X_{2}\right).\label{x1mdf}
\end{equation}

By using the definition for the performance parameter $x$(sec. \ref{steady}),
we obtain 
\begin{equation}
x=-q.\label{eq:xmdf}
\end{equation}
\end{proof}
The steady states that can be identified within a linear energy converter
correspond to the coupling of two observable processes, the so--called
spontaneous and other non-spontaneous. This distinction was not fully
addressed by Prigogine, since in his minimum entropy production principle
statement, he made no distinction about the nature of thermodynamic
forces in energy transfer or energy conversion processes. In this
way, we propose, by means of a simple step, the minimum dissipation
function theorem for linear energy converters in steady--state, the
energetic version of the ``constrained minimum entropy production
theorem''. In fact, each of the optimization criteria presented in
Section \ref{steady} satisfies a variational principle and can be
represented by specific flux--force relations \cite{SasaTasa06,Zupanivicetal04}.
Thus, energy conversion processes can be associated with different
stationary solutions. When the operation modes coincide, their steady
states are adjusted to one under the same initial condition (the same
degree of coupling $q$).

\subsection{Similar theorems to the ``constrained minimum entropy production
theorem'' for $\left(2\times2\right)$--energy converters}

As we have evidenced in Section \ref{steady}, in a large number of
non equilibrium steady--state processes that are carried out in both
living and man--made systems, appropriate constrains (boundary conditions)
can be found so that the energy conversion takes place in particular
operation modes. If the purpose in the operation of a $\left(2\times2\right)$--isothermal
linear energy converter is not to minimize the entropy production,
the steady states that dynamically constrain the processes within
the converter in the optimal performance regimes ($MPO$, $M\eta$,
$MEF$, $M\Omega$ and $MP\eta$) \cite{Stucki80,SantAriAng97,OdumPink55,Calvoetal01,KedemCapl65},
steady states revealed through the energy conversion theorems similar
to the Prigogine's one with constrained forces (\textbf{mdf--THEOREM}),
as well as their corresponding corollaries. The optimal regimes here
analyzed were taken from the field of Finite Time Thermodynamics (FTT)
\cite{AriPaeAng08,Angulo02}.

\begin{theorem} \textbf{\textsl{\large{}MPO-{}-THEOREM}} When a non
equilibrium steady--state system is characterized by two generalized
forces $X_{1}$ (the one associated with driven processes) and $X_{2}$
(another one associated with driver processes), it is constrained
to a particular steady state when the driver force is fixed. Then,
under the condition of \textbf{maximum power output}, the $J_{1}$
flux is equal to $-L_{11}X_{1}$. \end{theorem}
\begin{proof}
Let us consider the mathematical expression for power output {[}Eq.
(\ref{eq:potefic}){]}. The partial derivative of $P_{out}$ with
respect to $X_{1}$ assuming $X_{2}$ fixed, is: 
\begin{equation}
\left(\frac{\partial P_{O}}{\partial X_{1}}\right)_{X_{2}}=-T\left(J_{1}+L_{11}X_{1}\right).\label{dpo}
\end{equation}

By hypothesis $\left(\frac{\partial P_{O}}{\partial X_{1}}\right)_{X_{2}}=0$,
we get: 
\begin{equation}
J_{1}=-L_{11}X_{1}.\label{condparflux1mpo}
\end{equation}
\end{proof}
\begin{corollary}\textbf{ }\textbf{\textsl{MPO--COROLLARY}} If the
degree of coupling between the processes of a non equilibrium steady--state
system is measured by the coefficient $q$ and it is also operating
under the \textbf{MPO} regime, the cross effect between both generalized
forces, denoted by the parameter $x$, is given by $x=-\frac{q}{2}$.
\end{corollary}
\begin{proof}
From the \textbf{$MPO$--THEOREM}, the constraint $J_{1}=-L_{11}X_{1}$
lead us to write the force $X_{1}$ as: 
\begin{equation}
X_{1}=-\frac{q}{2}\left(\sqrt{\frac{L_{22}}{L_{11}}}X_{2}\right).\label{x1mpo}
\end{equation}

By using the definition for the performance parameter $x$(sec. \ref{steady}),
we have 
\begin{equation}
x=-\frac{q}{2}.\label{eq:xmpo}
\end{equation}
\end{proof}
\begin{theorem}{\large{} $M\eta$}\textbf{\textsl{\large{}-{}-THEOREM}}
When a non equilibrium steady--state system is characterized by two
generalized forces $X_{1}$ (the one associated with driven processes)
and $X_{2}$ (another one associated with driver processes), it is
arrested at a steady state when the driver force is fixed. Then, under
the condition of \textbf{maximum efficiency}, the $J_{1}$ flux is
equal to $-\left(\frac{1-\eta}{1+\eta}\right)L_{11}X_{1}$. \end{theorem}
\begin{proof}
Let us take the mathematical expression for efficiency {[}Eq. (\ref{eq:potefic}){]}.
By calculating the partial derivative with respect to $X_{1}$ assuming
$X_{2}$ fixed, we obtain 
\begin{equation}
\left(\frac{\partial\eta}{\partial X_{1}}\right)_{X_{2}}=-\frac{P_{O}}{P_{I}^{2}}\left(\frac{\partial P_{I}}{\partial X_{1}}\right)+\frac{1}{P_{I}}\left(\frac{\partial P_{O}}{\partial X_{1}}\right),\label{deta}
\end{equation}
where $P_{I}$ is the rate of incoming energy per temperature unit,
this amount of energy is associated with spontaneous flux. The derivative
$\left(\frac{\partial P_{O}}{\partial X_{1}}\right)$ is given in
\textbf{$MPO$--THEOREM}, while $\left(\frac{\partial P_{I}}{\partial X_{1}}\right)$
is: 
\begin{equation}
\left(\frac{\partial P_{I}}{\partial X_{1}}\right)=Tq\sqrt{L_{11}L_{22}}X_{2}.\label{dpi}
\end{equation}

Then, we can rewrite $\left(\frac{\partial\eta}{\partial X_{1}}\right)_{X_{2}}$
as follows: 
\begin{equation}
\left(\frac{\partial\eta}{\partial X_{1}}\right)_{X_{2}}=-\frac{T}{P_{I}}\left[\frac{P_{O}}{P_{I}}q\sqrt{L_{11}L_{22}}X_{2}+J_{1}+L_{11}X_{1}\right].\label{detapo}
\end{equation}

By hypothesis $\left(\frac{\partial\eta}{\partial X_{1}}\right)_{X_{2}}=0$
and, by using the definition of $\eta$, we get: 
\begin{equation}
\eta\left(q\sqrt{L_{11}L_{22}}X_{2}+L_{11}X_{1}-L_{11}X_{1}\right)+J_{1}+L_{11}X_{1}=0.
\end{equation}

Finally, 
\begin{equation}
J_{1}=-\left(\frac{1-\eta}{1+\eta}\right)L_{11}X_{1}.\label{condparflux1me}
\end{equation}
\end{proof}
\begin{corollary}{\large{} $M\eta$}\textbf{\textsl{\large{}-{}-COROLLARY}}
If the degree of coupling between the processes of a non equilibrium
steady--state system is measured by the coefficient $q$ and it is
also operating under the \textbf{$M\eta$} regime, the cross effect
between both generalized forces, denoted by the parameter $x$, is
given by $x=-\frac{q}{1+\sqrt{1-q^{2}}}$. \end{corollary}
\begin{proof}
From the \textbf{$M\eta$--THEOREM}, the constraint $J_{1}=-\left(\frac{1-\eta}{1+\eta}\right)L_{11}X_{1}$
lead us to write the force $X_{1}$ as: 
\begin{equation}
X_{1}=-\frac{q\left(1+\eta\right)}{2}\left(\sqrt{\frac{L_{22}}{L_{11}}}X_{2}\right).\label{x1eta}
\end{equation}

By using the definition for $x$, as well as Eq.\ref{eq:potefic}c,
we have 
\begin{equation}
x=-\frac{q}{1+\sqrt{1-q^{2}}}.\label{eq:xme}
\end{equation}
\end{proof}

\subsubsection{Energy conversion theorems for trade-off functions}

\begin{theorem}{\large{} $MEF$}\textbf{\textsl{\large{}-{}-THEOREM}}
When a non equilibrium steady--state system is characterized by two
generalized forces $X_{1}$ (the one associated with driven processes)
and $X_{2}$ (another one associated with driver processes), it is
arrested at a steady state when the driver force is fixed. Then, under
the condition of \textbf{maximum ecological function}, the $J_{1}$
flux is equal to $-\frac{1}{3}L_{11}X_{1}$. \end{theorem}
\begin{proof}
Let us consider the mathematical expression for the so-called ecological
function {[}Eq. (\ref{eq:compfunc}a){]}. The partial derivative with
respect to $X_{1}$ assuming $X_{2}$ fixed, is: 
\begin{equation}
\left(\frac{\partial E_{F}}{\partial X_{1}}\right)_{X_{2}}=-T\left(J_{1}+L_{11}X_{1}+2J_{1}\right).\label{def}
\end{equation}

By hypothesis $\left(\frac{\partial E_{F}}{\partial X_{1}}\right)_{X2}=0$,
we get 
\begin{equation}
J_{1}=-\frac{1}{3}L_{11}X_{1}.\label{condparflux1mef}
\end{equation}
\end{proof}
\begin{corollary}{\large{} $MEF$}\textbf{\textsl{\large{}-{}-COROLLARY}}
If the degree of coupling between the processes of a non equilibrium
steady--state system is measured by the coefficient $q$ and it is
also operating under the \textbf{MEF} regime, the cross effect between
both generalized forces, denoted by the parameter $x$, is given by
$x=-\frac{3q}{4}$. \end{corollary}
\begin{proof}
From the \textbf{$MEF$--THEOREM}, the constraint $J_{1}=-\frac{1}{3}L_{11}X_{1}$
lead us to write the force $X_{1}$ as: 
\begin{equation}
X_{1}=-\frac{3q}{4}\left(\frac{L_{22}}{L_{11}}X_{2}\right).\label{x1ef}
\end{equation}
By using the definition for $x$, we have 
\begin{equation}
x=-\frac{3q}{4}.\label{eq:xmef}
\end{equation}
\end{proof}
\begin{theorem} {\large{}$M\Omega$}\textbf{\textsl{\large{}-{}-THEOREM}}
When a non equilibrium steady--state system is characterized by two
generalized forces $X_{1}$ (the one associated with driven processes)
and $X_{2}$ (another one associated with driver processes), it is
arrested at a steady state when the driver force is fixed. Then, under
the condition of \textbf{maximum omega function}, the $J_{1}$ flux
is equal to $-\left(\frac{2-\eta_{M}}{2+\eta_{M}}\right)L_{11}X_{1}$.
\end{theorem}
\begin{proof}
Let us take the mathematical expression for the so-called omega function
{[}Eq. (\ref{eq:compfunc}b){]}. By calculating the partial derivative
of $\Omega$ with respect to $X_{1}$ and by assuming $X_{2}$ fixed,
\begin{equation}
\left(\frac{\partial\Omega}{\partial X_{1}}\right)_{X_{2}}=-T\left[\left(2+\eta_{M}\right)J_{1}+\left(2-\eta_{M}\right)L_{11}X_{1}\right].\label{do}
\end{equation}

By hypothesis $\left(\frac{\partial\Omega}{\partial X_{1}}\right)_{X_{2}}=0$.
Then, we obtain 
\begin{equation}
J_{1}=-\left(\frac{2-\eta_{M}}{2+\eta_{M}}\right)L_{11}X_{1}.\label{eq:fluxcondmof}
\end{equation}
\end{proof}
\begin{corollary}{\large{} $M\Omega$}\textbf{\textsl{\large{}-{}-COROLLARY}}
If the degree of coupling between the processes of a non equilibrium
steady--state system is measured by the coefficient $q$ and it is
also operating under the \textbf{$M\Omega$} regime, the cross effect
between both generalized forces, denoted by the parameter $x$, is
given by $x=-\frac{q\left(4-q^{2}+4\sqrt{1-q^{2}}\right)}{4\left(1+\sqrt{1-q^{2}}\right)^{2}}$.
\end{corollary}
\begin{proof}
From the \textbf{$M\Omega$--THEOREM}, the constraint $J_{1}=-\left(\frac{2-\eta_{M}}{2+\eta_{M}}\right)L_{11}X_{1}$
lead us to write the force $X_{1}$ as: 
\begin{equation}
X_{1}=-\frac{q\left(2+\eta_{M}\right)}{4}\left(\sqrt{\frac{L_{22}}{L_{11}}}X_{2}\right),\label{x1o}
\end{equation}
by using the definition for $x$ and $\eta_{M}=\eta\left(x^{M\eta},q\right)$,
we obtain

\begin{equation}
x=-\frac{q\left(4-q^{2}+4\sqrt{1-q^{2}}\right)}{4\left(1+\sqrt{1-q^{2}}\right)^{2}}.\label{eq:xmo}
\end{equation}
\end{proof}
\begin{theorem} {\large{}$MP_{\eta}$}\textbf{\textsl{\large{}-{}-THEOREM}}
When a non equilibrium steady--state system is characterized by two
generalized forces $X_{1}$ (the one associated with driven processes)
and $X_{2}$ (another one associated with driver processes), it is
arrested at a steady state when the driver force is fixed. Then, under
the condition of \textbf{maximum efficient power}, the $J_{1}$ flux
is equal to $-\left(\frac{2-\eta}{2+\eta}\right)L_{11}X_{1}$. \end{theorem}
\begin{proof}
Let us consider the mathematical expression for the so-called efficient
power {[}Eq. (\ref{eq:compfunc}c){]}. By calculating the partial
derivative of $P_{\eta}$ with respect to $X_{1}$ assuming $X_{2}$
fixed, 
\begin{equation}
\left(\frac{\partial P_{\eta}}{\partial X_{1}}\right)_{X_{2}}=-\frac{TP_{O}}{P_{I}}\left[\frac{P_{O}}{P_{I}}L_{12}X_{2}+2\left(J_{1}+L_{11}X_{1}\right)\right].\label{dpe}
\end{equation}

By hypothesis $\left(\frac{\partial P_{\eta}}{\partial X_{1}}\right)_{X_{2}}=0$
and by using the definition of $\eta$, we have 
\begin{equation}
2J_{1}+2L_{11}X_{1}+\eta\left(L_{12}X_{2}+L_{11}X_{1}-L_{11}X_{1}\right)=0.\label{x1pe}
\end{equation}

Finally, 
\begin{equation}
J_{1}=-\frac{2-\eta}{2+\eta}L_{11}X_{1}.\label{condparflux1mpe}
\end{equation}
\end{proof}
\begin{corollary} {\large{}$MP_{\eta}$}\textbf{\textsl{\large{}--COROLLARY}}
If the degree of coupling between the processes of a non equilibrium
steady--state system is measured by the coefficient $q$ and it is
also operating under the \textbf{MP$\eta$} regime, the cross effect
between both generalized forces, denoted by the parameter $x$, is
given by $x=-\frac{4+q^{2}-\sqrt{q^{4}-16q^{2}+16}}{6q}$. \end{corollary}
\begin{proof}
From the \textbf{MP$\eta$--THEOREM}, the constraint $J_{1}=-\frac{2-\eta}{2+\eta}L_{11}X_{1}$
lead us to write the force $X_{1}$ as: 
\begin{equation}
X_{1}=-\frac{q\left(2+\eta\right)}{4}\left(\sqrt{\frac{L_{22}}{L_{11}}}X_{2}\right).\label{x1mpe}
\end{equation}

By using the definitions for $x$ and $\eta$, we get 
\begin{equation}
x=-\frac{4+q^{2}-\sqrt{q^{4}-16q^{2}+16}}{6q}.\label{eq:xmep}
\end{equation}
\end{proof}
As has been pointed out in some works \cite{Castilloetal97,Hoover02},
a boundary condition that at the same time is linked to a particular
dynamic performance mode has several stable steady solutions. The
energetic optimization criteria presented in Section 2 can be equivalent
when the degree of coupling between the flows adopt particular values.
For instance, the steady state associated with the minimum dissipation
function condition is identical to the other extreme criteria, in
Table \ref{tab:mdfboundcond} the values of $q$ are displayed. Other
constraints can also be established using the previous corollaries
and evaluating $q$ in the characteristic functions presented in Section
2.

\begin{table}
\caption{\label{tab:mdfboundcond}Physical constraints on linear energy converters
so that two or more steady states coincide with the minimum dissipation
function condition.}

\centering{}%
\begin{tabular}{c|c}
\hline 
Boundary Conditions  & $q$\tabularnewline
\hline 
$\Phi_{mdf}\left(q\right)=\Phi_{MPO}\left(q\right)$  & $0$\tabularnewline
\hline 
$\Phi_{mdf}\left(q\right)=\Phi_{M\eta}\left(q\right)$  & $0$ and $1$\tabularnewline
\hline 
$\Phi_{mdf}\left(q\right)=\Phi_{MEF}\left(q\right)$  & $0$\tabularnewline
\hline 
$\Phi_{mdf}\left(q\right)=\Phi_{M\Omega}\left(q\right)$  & $0$ \tabularnewline
\hline 
$\Phi_{mdf}\left(q\right)=\Phi_{MP\eta}\left(q\right)$  & $0$ \tabularnewline
\hline 
\end{tabular}
\end{table}

All of the above energetic trade--offs fulfill the condition $\Phi>0$,
since the non-equilibrium steady states associated with them, only
appear when external thermodynamic forces are linked under an extreme
condition. Thus, the operation modes involve the maintenance of steady
states for coupled processes that waste free energy at different rates.

\subsection{Temporal evolution of the characteristic functions and stability
of their steady states}

The analysis of external perturbations on different types of thermal
cycle models has been topic of interest \cite{Santillan01,Guzman05,Huang08,Wouagfack17}.
Since thermal engines operate with many cycles per unit time, the
effect of noisy perturbations has been able to help to have a well-design
systems guaranteeing stability in their operating regimes (steady-state
regimes) \cite{Gonzalez20,Lee20,Valencia21}. In several analysis
on LIT \cite{KondePrigo98,DeGroot62,Glansdorff54,Glansdorff71,Nicolis77},
it has been shown that spontaneous fluctuations on a particular steady
state drive the system back to its condition of \textit{minimum entropy
production. }In this subsection we show that the stability of linear
steady--states can be also attained for the constrained cases.

\subsubsection{Dissipation function temporal evolution}

For the case of a $(2\times2)$--isothermal energy converters, let
us consider the dissipation function $\Phi=\Phi(X_{i},J_{i})$ in
a non-equilibrium steady state. As $\Phi=T\int_{V}\left(J_{1}X_{1}+J_{2}X_{2}\right)\,dV$,
the time variation of $\Phi$ can be written as: 
\begin{eqnarray}
\frac{d\Phi}{dt} & = & \int_{V}\left(J_{1}\frac{dX_{1}}{dt}+J_{2}\frac{dX_{2}}{dt}\right)\,dV+\nonumber \\
 &  & \int_{V}\left(X_{1}\frac{dJ_{1}}{dt}+X_{2}\frac{dJ_{2}}{dt}\right)\,dV\nonumber \\
 & \equiv & \frac{d_{X}\Phi}{dt}+\frac{d_{J}\Phi}{dt},
\end{eqnarray}
where $\nicefrac{d_{X}\Phi}{dt}$ is the time variation of thermodynamic
forces associated with spontaneous and non-spontaneous fluxes, and
$\nicefrac{d_{J}\Phi}{dt}$ is the temporal change of these conjugated
fluxes. In the linear regime a general property has been stated \cite{KondePrigo98,DeGroot62,Glansdorff54}:
\begin{equation}
\frac{d_{X}\Phi}{dt}=\frac{d_{J}\Phi}{dt}.\label{eq:cond1}
\end{equation}

If we assume hereinafter homogeneity and unitary volume, the time
derivative of each $X_{i}$ can be written by noting that 
\begin{equation}
\frac{dX_{i}}{dt}=\sum_{k}\left(\frac{\partial X_{i}}{\partial a_{k}}\right)\frac{da_{k}}{dt}.\label{eq:cond2}
\end{equation}
Thus, we find 
\begin{equation}
\frac{d\Phi}{dt}=2T\left[J_{1}\sum_{k=1}^{2}\left(\frac{\partial X_{1}}{\partial a_{k}}\right)J_{k}+J_{2}\sum_{k=1}^{2}\left(\frac{\partial X_{2}}{\partial a_{k}}\right)J_{k}\right],
\end{equation}
where Onsager reciprocal relations and the definition of $\nicefrac{da_{k}}{dt}\equiv J_{k}$
have been considered. On the other side, since the LIT has been constructed
using the so--called Local Equilibrium Hypothesis \cite{Onsager31I,ColinGolds03},
in a neighborhood of the equilibrium state, $\nicefrac{\partial X_{i}}{\partial a_{k}}$
would retain its negative definiteness for the spontaneous terms,
while non-spontaneous ones are positive. In addition, because of $X_{1}<0$
and $X_{2}>0$. Hence, in this neighborhood we have 
\begin{equation}
\frac{d\Phi}{dt}<0.\label{eq:mdfstab}
\end{equation}

This condition ensures the stability of the steady state associated
with the dissipation function, since in general the magnitude of the
direct effects are greater than the cross effects.

\subsubsection{Power output temporal evolution}

In the same way, given an arbitrary volume, let us now consider the
power output for a $(2\times2)$-isothermal energy converter $P_{O}=P_{O}(X_{i},J_{i})$.
Then, the time variation of $P_{O}$ is 
\begin{equation}
\frac{dP_{O}}{dt}\equiv-\left(\frac{d_{X}P_{O}}{dt}+\frac{d_{J}P_{O}}{dt}\right),\label{eq:tempotO}
\end{equation}
under the same mathematical assumptions of \cite{KondePrigo98,DeGroot62,Glansdorff54},
\begin{equation}
\frac{d_{J}P_{O}}{dt}<\frac{d_{X}P_{O}}{dt},\label{eq:cond3}
\end{equation}
and by using the condition given by Eq. (\ref{eq:cond2}). We have:
\begin{eqnarray}
\frac{dP_{O}}{dt} & = & -T\left[\left(J_{1}+L_{11}X_{1}\right)\sum_{k=1}^{2}\left(\frac{\partial X_{1}}{\partial a_{k}}\right)J_{k}+\right.\nonumber \\
 &  & \left.\left(J_{2}-L_{22}X_{2}\right)\sum_{k=1}^{2}\left(\frac{\partial X_{2}}{\partial a_{k}}\right)J_{k}\right].\label{eq:derMPO}
\end{eqnarray}
Therefore, in the neighborhood of the equilibrium state 
\begin{equation}
\frac{dP_{O}}{dt}<0,\label{eq:MPOstab}
\end{equation}
since $|L_{11}X_{1}|<|L_{22}X_{2}|$. This new constrain guarantees
the stability of a steady state linked to the power output regime.

\subsubsection{Efficiency temporal evolution}

For the case of the efficiency in this type of linear energy converters,
$\eta=\eta(X_{i},J_{i})$. The temporal evolution of $\eta$ can be
written as 
\begin{equation}
\frac{d\eta}{dt}\equiv-\frac{1}{P_{I}^{2}}\left[P_{I}\left(\frac{d_{X}P_{O}}{dt}+\frac{d_{J}P_{O}}{dt}\right)-P_{O}\left(\frac{d_{X}P_{I}}{dt}+\frac{d_{J}P_{I}}{dt}\right)\right];\label{eq:stabMet}
\end{equation}
analogously, using the same mathematical constraints, \begin{subequations}
\label{eq:cond4} 
\begin{eqnarray}
|P_{O}| & < & |P_{I}|\\
\frac{d_{X}P_{I}}{dt} & < & \frac{d_{J}P_{I}}{dt},
\end{eqnarray}
\end{subequations} and by considering Eq. (\ref{eq:cond2}), 
\begin{eqnarray}
\frac{d\eta}{dt} & = & -\frac{1}{P_{I}^{2}}\left[P_{I}\left(\frac{dP_{O}}{dt}\right)-\right.\nonumber \\
 &  & \left.P_{O}\left\lbrace \left(J_{1}-L_{11}X_{1}\right)\sum_{k=1}^{2}\left(\frac{\partial X_{1}}{\partial a_{k}}\right)J_{k}\right.\right.+\nonumber \\
 &  & \left.\left.\left(J_{2}+L_{22}X_{2}\right)\sum_{k=1}^{2}\left(\frac{\partial X_{2}}{\partial a_{k}}\right)J_{k}\right\rbrace \right].\label{eq:derMet}
\end{eqnarray}
Hence, in the neighborhood of the equilibrium state: 
\begin{equation}
\frac{d\eta}{dt}<0,\label{eq:Metstab}
\end{equation}
where $|L_{11}X_{1}|<|L_{22}X_{2}|$ must be fulfilled. Once again,
this new constraint describes the stability of a steady state related
to this energetic function.

\subsubsection{Temporal evolution of trade-off functions}

Finally, let us take the mathematical expressions of the three objective
functions: ecological function, omega function and efficient power;
in order to analyze the effect of the fluctuations around their respective
stable state.

\paragraph{Temporal evolution of $E_{F}$.}

For this $(2\times2)$-linear system, $E_{F}=E_{F}(X_{i},J_{i})$
can be studied when it is disturbed, 
\begin{equation}
\frac{dE_{F}}{dt}=-T\left[2\left(\frac{d_{X}P_{O}}{dt}+\frac{d_{J}P_{O}}{dt}\right)+\frac{d_{X}P_{I}}{dt}+\frac{d_{J}P_{I}}{dt}\right].\label{eq:stabeco}
\end{equation}
From Eqs. (\ref{eq:cond1}) and (\ref{eq:cond3}), we can state that
\begin{equation}
\frac{d_{J}E_{F}}{dt}<\frac{d_{X}E_{F}}{dt},\label{eq:cond5}
\end{equation}
and by using the definition of Eq. (\ref{eq:cond2}) as well as Eqs.
(\ref{eq:mdfstab}) and (\ref{eq:derMPO}), 
\begin{eqnarray}
\frac{dE_{F}}{dt} & = & -T\left[\left(3J_{1}+L_{11}X_{1}\right)\sum_{k=1}^{2}\left(\frac{\partial X_{1}}{\partial a_{k}}\right)J_{k}\right.+\nonumber \\
 &  & \left.\left(3J_{2}-L_{22}X_{2}\right)\sum_{k=1}^{2}\left(\frac{\partial X_{2}}{\partial a_{k}}\right)J_{k}\right].\label{eq:derMeco}
\end{eqnarray}
Then, in the vicinity of the equilibrium state, 
\begin{equation}
\frac{dE_{F}}{dt}<0.\label{eq:Mecsatb}
\end{equation}

\paragraph{Temporal evolution of $\Omega$.}

Now, let us consider the omega function $\Omega=\Omega(X_{i},J_{i})$,
its temporal variation is: 
\begin{equation}
\frac{d\Omega}{dt}=-T\left[2\left(\frac{d_{X}P_{O}}{dt}+\frac{d_{J}P_{O}}{dt}\right)+\eta_{M}\left(\frac{d_{X}P_{I}}{dt}+\frac{d_{J}P_{I}}{dt}\right)\right].\label{eq:stabome}
\end{equation}
Eqs. (\ref{eq:cond1}) and (\ref{eq:cond3}) allow us to establish
that 
\begin{equation}
\frac{d_{J}\Omega}{dt}<\frac{d_{X}\Omega}{dt},\label{eq:cond6}
\end{equation}
then, taking the expressions given by Eqs. (\ref{eq:cond2}) and (\ref{eq:derMPO}),
\begin{eqnarray}
\frac{d\dot{\Omega}}{dt} & = & -T\left[\left\lbrace \left(2+\eta_{M}\right)J_{1}+\left(2-\eta_{M}\right)L_{11}X_{1}\right\rbrace \sum_{k=1}^{2}\left(\frac{\partial X_{1}}{\partial a_{k}}\right)J_{k}+\right.\nonumber \\
 &  & \left.\left\lbrace \left(2+\eta_{M}\right)J_{2}-\left(2-\eta_{M}\right)L_{22}X_{2}\right\rbrace \sum_{k=1}^{2}\left(\frac{\partial X_{2}}{\partial a_{k}}\right)J_{k}\right].\label{eq:Momstab}
\end{eqnarray}
Once again, in the zone near to the equilibrium state 
\begin{equation}
\frac{d\Omega}{dt}<0.\label{eq:derMome}
\end{equation}

\paragraph{Temporal evolution of $P_{\eta}$.}

Finally, let us use the mathematical expression for the efficient
power $P_{\eta}=P_{\eta}(X_{i},J_{i})$, the temporal evolution of
this objective function is: 
\begin{equation}
\frac{dP_{\eta}}{dt}=\frac{TP_{O}}{P_{I}^{2}}\left[2P_{I}\left(\frac{d_{X}P_{O}}{dt}+\frac{d_{J}P_{O}}{dt}\right)-P_{O}\left(\frac{d_{X}P_{I}}{dt}+\frac{d_{J}P_{I}}{dt}\right)\right].\label{eq:stabmpe}
\end{equation}
Due to the validity of Eqs. (\ref{eq:cond1}) and (\ref{eq:cond3}),
we can state that 
\begin{equation}
\frac{d_{J}P_{\eta}}{dt}<\frac{d_{X}P_{\eta}}{dt},\label{eq:cond7}
\end{equation}
by taking into account the constraints {[}Eqs. (\ref{eq:derMPO})
and (\ref{eq:cond4}){]}, in the neighborhood of the equilibrium state:
\begin{equation}
\frac{dP_{\eta}}{dt}<0.\label{eq:dermpe}
\end{equation}

Inequalities (\ref{eq:derMeco}), (\ref{eq:derMome}) and (\ref{eq:dermpe})
express new constraints that describe the stability of three steady
states associated with three different objective functions.

The above-mentioned conditions ensure the stability of isothermal
energy converters, whose energetic processes occur in the linear regime
when small perturbations are considered. As the so-called characteristic
functions are positive and also their respective temporal variation
are negative, constitute the so-called Lyapunov conditions \cite{Glansdorff71}.
They guarantee the stability of any dynamic state, i.e, those states
are considered simple attractors when the systems experience energetic
fluctuations.

\section{An application of energy theorems: Electric circuit with resistors
elements and two coupled fluxes}

In this section, we apply the previous results, associated with steady
states without minimum entropy production, to study the energetics
for a $\left(2\times2\right)$--isothermal linear energy converter
consisting of an electric circuit of two meshes with passive elements
(resistors) and powered by two DC voltage sources (see Fig. \ref{fig:schemcircR}).
Through Kirchhoff's equations that describe the energy conservation
between the meshes, within the context of generalized flows and forces,
we found the so-called cross effect $\left(L_{12}=L_{21}\right)$
is fulfilled. Then, from the parameters related to its design ($q$)
and its operation modes ($x$), we can characterize all the optimal
operating criteria that correspond to particular stationary states.

\begin{figure}
\begin{centering}
\includegraphics[scale=0.75]{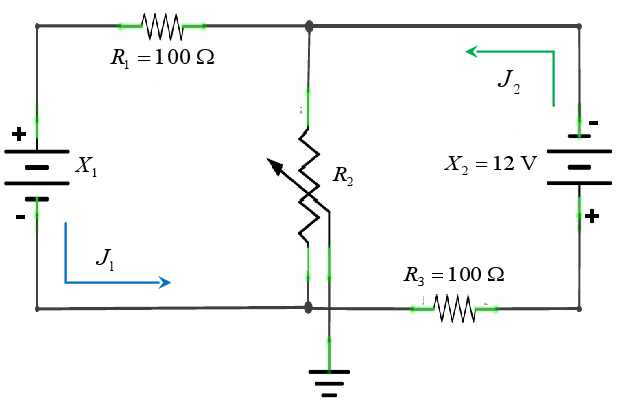} 
\par\end{centering}
\caption{\label{fig:schemcircR}Two-mesh resistive electric circuit, modeled
as a $(2\times2)$-simple isothermal energy converter. In this case,
the voltage source $X_{1}$ is associated with the non-spontaneous
flux $J_{1}$, while the fixed voltage source $X_{2}$ is related
to the spontaneous flux $J_{2}$.}
\end{figure}

\subsection{Dynamic equations of the resistive circuit}

Let us consider the electric currents involved are time independent,
with $J_{1}$ the driven flux that flows through the mesh 1 and $J_{2}$
the driver flux that gets about through the mesh 2. In addition, $X_{1}$
and $X_{2}$ are the DC voltage sources that promote the fluxes. As
the circuit phenomenological equations can be associated with the
Kirchhoff laws, then for the mesh 1 (left side in figure \ref{fig:schemcircR})
we have, 
\begin{equation}
X_{1}=\left(R_{1}+R_{2}\right)J_{1}-R_{2}J_{2},\label{eq:kirchh1}
\end{equation}
while for the mesh 2 (right side of Fig. \ref{fig:schemcircR}), 
\begin{equation}
X_{2}=\left(R_{2}+R_{3}\right)J_{2}-R_{2}J_{1}.\label{eq:kirchh2}
\end{equation}

In order to study the phenomenological equation of this system in
the Onsager context, we must reverse Eqs. (\ref{eq:kirchh1}) and
(\ref{eq:kirchh2}), i.e, writing $\left(J_{1},J_{2}\right)$ in terms
of $\left(X_{1},X_{2}\right)$. Hence, 
\begin{equation}
\left(\begin{array}{c}
J_{1}\\
J_{2}
\end{array}\right)=\left(\begin{array}{cc}
\frac{R_{2}+R_{3}}{\Delta} & \frac{R_{2}}{\Delta}\\
\frac{R_{2}}{\Delta} & \frac{R_{1}+R_{2}}{\Delta}
\end{array}\right)\left(\begin{array}{c}
X_{1}^{\prime}\\
X_{2}^{\prime}
\end{array}\right),\label{eq:kirchhinver}
\end{equation}
where $\Delta=R_{1}\left(R_{2}+R_{3}\right)+R_{2}R_{3}$ and $X_{1,2}^{\prime}=\nicefrac{X_{1,2}}{T}$.
As it can be noticed the equation system {[}Eq. (\ref{eq:kirchhinver}){]}
is by itself the generalized Onsager equations whose equality in the
crossed coefficients represents the contribution of each force $X_{i}$
with the flux $J_{k\neq i}$. In this case, the parameter $q$ is
related to the intrinsic features of the circuit, and can be rewritten
in terms of the resistive elements nominal values, as follows: 
\begin{equation}
q=\frac{L_{12}}{\sqrt{L_{11}L_{22}}}=\frac{R_{2}}{\sqrt{\Delta+R_{2}^{2}}},\label{eq:couplwitrs}
\end{equation}
and the force ratio $x$ which has all the extra thermodynamic information
of the system, it takes the mathematical expression: 
\begin{equation}
x=\sqrt{\frac{L_{11}}{L_{22}}}\frac{X_{1}}{X_{2}}=\sqrt{\frac{R_{2}+R_{3}}{R_{1}+R_{2}}}\frac{X_{1}}{X_{2}}.\label{eq:forcratwitrs}
\end{equation}

In the following, we will study every optimization criteria in terms
of arbitrary values that the resistors can adopt, with the aim not
only of characterizing the steady state related to an optimal coupling
in the circuit stable configuration, but also of ensuring the system
adapts to the selected operating regime.

\subsection{Steady state of minimum dissipation function ($mdf$)}

From the \textbf{$mdf$--THEOREM} it follows that $x_{mdf}$ can
be rewritten by replacing Eq. (\ref{eq:couplwitrs}) in Eq. (\ref{eq:forcratwitrs}),
as 
\begin{equation}
x_{mdf}=-\frac{R_{2}}{\sqrt{\Delta+R_{2}^{2}}},\label{eq:xmdfwitrs}
\end{equation}
then, the force $X_{1}$ is rewritten in this regime as: 
\begin{equation}
X_{1}^{mdf}=-\frac{R_{2}}{R_{2}+R_{3}}X_{2}.\label{eq:x1drivmdf}
\end{equation}

Thus, the energetics under $mdf$ operation regime is, \begin{subequations}
\label{eq:mdfregrs} 
\begin{eqnarray}
\Phi_{mdf} & = & \frac{TX_{2}^{2}}{R_{2}+R_{3}}\\
P_{mdf} & = & 0\\
\eta_{mdf} & = & 0.
\end{eqnarray}
\end{subequations}

\subsection{Steady state of maximum power output ($MPO$)}

By applying the \textbf{$MPO$--THEOREM}, we find that $x_{MPO}$
can be written as: 
\begin{equation}
x_{MPO}=-\frac{R_{2}}{2\sqrt{\Delta+R_{2}^{2}}},\label{eq:xmpowitrs}
\end{equation}
thereby, the force $X_{1}$ acquires the form 
\begin{equation}
X_{1}^{MPO}=-\frac{R_{2}}{2\left(R_{2}+R_{3}\right)}X_{2}.\label{eq:x1drivmpo}
\end{equation}

The energetics beneath $MPO$ operation regime is: \begin{subequations}
\label{eq:mporegrs} 
\begin{eqnarray}
\Phi_{MPO} & = & \left\{ \frac{4\Delta+R_{2}^{2}}{4\left(R_{2}+R_{3}\right)\Delta}\right\} TX_{2}^{2}\\
P_{MPO} & = & \left\{ \frac{R_{2}^{2}}{4\left(R_{2}+R_{3}\right)\Delta}\right\} TX_{2}^{2}\\
\eta_{MPO} & = & \frac{R_{2}^{2}}{4\Delta+2R_{2}^{2}}.
\end{eqnarray}
\end{subequations}

\subsection{Steady state of maximum efficiency ($M\eta$)}

By taking the \textbf{$M\eta$--THEOREM}, we have $x_{M\eta}$ in
terms of the resistors values, 
\begin{equation}
x_{M\eta}=-\frac{R_{2}}{\sqrt{\Delta+R_{2}^{2}}+\sqrt{\Delta}},\label{eq:xmeficwitrs}
\end{equation}
the force associated with driven flux $X_{1}$ takes the form:

\begin{equation}
X_{1}^{M\eta}=-\frac{R_{2}}{\left(R_{2}+R_{3}\right)+\sqrt{\frac{\Delta}{R_{1}+R_{2}}}}X_{2}.\label{eq:x1drivmefic}
\end{equation}

Therefore, the energetics that results in the $M\eta$ regime, \begin{subequations}
\label{eq:meficregrs} 
\begin{eqnarray}
\Phi_{M\eta} & = & \frac{2TX_{2}^{2}}{\left(R_{2}+R_{3}\right)\left(1+\sqrt{\Gamma}\right)}\\
P_{M\eta} & = & \frac{TX_{2}^{2}R_{2}^{2}}{\sqrt{\left(R_{2}+R_{3}\right)\left(\Delta+R_{2}^{2}\right)\Delta}\left(1+\sqrt{\Gamma}\right)^{2}}\\
\eta_{M\eta} & = & \frac{R_{2}^{2}\sqrt{C}}{2\left(1+\sqrt{\Gamma}\right)\left(R_{1}+R_{3}\right)+R_{2}^{2}\sqrt{\Gamma}},
\end{eqnarray}
\end{subequations} with $\sqrt{\Gamma}=\sqrt{\frac{\Delta}{\Delta+R_{2}^{2}}}$.

\subsection{Steady state of maximum ecological function ($MEF$)}

From the \textbf{$MEF$--THEOREM}, it can be shown that $x_{MEF}$
in terms of the circuit elements {[}Eq. (\ref{eq:couplwitrs}){]}
is, 
\begin{equation}
x_{MEF}=-\frac{3R_{2}}{4\sqrt{\Delta+R_{2}^{2}}},\label{eq:xmefwitrs}
\end{equation}
then, the force $X_{1}$ associated with the non-spontaneous flux,
\begin{equation}
X_{1}^{MEF}=-\frac{3R_{2}}{4\left(R_{2}+R_{3}\right)}X_{2}.\label{eq:x1drivmef}
\end{equation}

In such a way, the energetics corresponding to the $MEF$ regime remains,
\begin{subequations} \label{eq:mefregrs} 
\begin{eqnarray}
\Phi_{MEF} & = & \left\{ \frac{16\Delta+R_{2}^{2}}{16\left(R_{2}+R_{3}\right)\Delta}\right\} TX_{2}^{2}\\
P_{MEF} & = & \left\{ \frac{3R_{2}^{2}}{16\left(R_{2}+R_{3}\right)\Delta}\right\} TX_{2}^{2}\\
\eta_{MEF} & = & \frac{3R_{2}^{2}}{4\Delta+R_{2}^{2}}.
\end{eqnarray}
\end{subequations}

\subsection{Steady state of maximum omega function ($M\Omega$)}

By using the \textbf{$M\Omega$--THEOREM}, it is derived that $x_{M\Omega}$
can also be denoted as 
\begin{equation}
x_{M\Omega}=-\frac{R_{2}\left(3+\Gamma+4\sqrt{\Gamma}\right)}{4\sqrt{\Delta+R_{2}^{2}}\left(1+\sqrt{\Gamma}\right)^{2}},\label{eq:xmomewitrs}
\end{equation}
in order to characterize the force $X_{1}$, 
\begin{equation}
X_{1}^{M\Omega}=-\frac{R_{2}\left(3+\Gamma+4\sqrt{\Gamma}\right)}{4\left(R_{2}+R_{3}\right)\left(1+\sqrt{\Gamma}\right)^{2}}X_{2}\label{eq:x1drivmomf}
\end{equation}

Thus, the energetics is evaluated in the $M\Omega$ regime: \begin{widetext}
\begin{subequations} \label{eq:momefregrs} 
\begin{eqnarray}
\Phi_{M\Omega} & = & \frac{\left(R_{1}+R_{2}\right)TX_{2}^{2}}{\Delta}\left\{ \sqrt{\Gamma}+\left(1-\Gamma\right)\left[\frac{\left(3+\Gamma+4\sqrt{\Gamma}\right)^{2}}{16\left(1+\sqrt{\Gamma}\right)^{4}}-\frac{1}{2}\right]\right\} \\
P_{M\Omega} & = & \frac{TX_{2}^{2}R_{2}^{2}\left(3+\Gamma+4\sqrt{\Gamma}\right)\left(1+3\Gamma+4\sqrt{\Gamma}\right)}{16\left(R_{2}+R_{3}\right)\Delta\left(1+\sqrt{\Gamma}\right)^{4}}\\
\eta_{M\Omega} & = & \frac{R_{2}^{2}\left(3+\Gamma+4\sqrt{\Gamma}\right)\left(1+3\Gamma+4\sqrt{\Gamma}\right)}{8\left(\Delta+R_{2}^{2}\right)\left(1+\sqrt{\Gamma}\right)^{4}\left[1+\Gamma-\sqrt{\Gamma}\right]}.
\end{eqnarray}
\end{subequations} \end{widetext}

\subsection{Steady state of maximum efficient power ($MP\eta$)}

Finally, by means of the \textbf{MP$\eta$--THEOREM}, the forces
ratio $x_{MP\eta}$ can be rewritten by using Eq. (\ref{eq:couplwitrs}),
\begin{equation}
x_{MP\eta}=-\frac{5-\Gamma-\sqrt{16\Gamma+\left(1-\Gamma\right)^{2}}}{6q},\label{eq:xmepwitrs}
\end{equation}
from which it is obtained that $X_{1}$ is, 
\begin{equation}
X_{1}^{MP\eta}=-\frac{\left(R_{1}+R_{2}\right)\left[5-\Gamma-\sqrt{16\Gamma+\left(1-\Gamma\right)^{2}}\right]}{6R_{2}}X_{2}.\label{eq:x1drivmep}
\end{equation}

Accordingly, the energetics of the system in the $MP\eta$ regime,
\begin{widetext}
\begin{subequations} \label{eq:mepregrs} 
\begin{eqnarray}
\Phi_{MP\eta} & = & \frac{\left(R_{1}+R_{2}\right)TX_{2}^{2}}{\Delta}\left\{ 1+\frac{\left(\Delta+R_{2}^{2}\right)\left[5-\Gamma-\sqrt{16\Gamma+\left(1-\Gamma\right)^{2}}\right]^{2}}{36R_{2}^{2}}-\frac{\left[5-\Gamma-\sqrt{16\Gamma+\left(1-\Gamma\right)^{2}}\right]}{3}\right\} \\
P_{MP\eta} & = & \frac{TX_{2}^{2}\left(R_{1}+R_{2}\right)\left[5-\Gamma-\sqrt{16\Gamma+\left(1-\Gamma\right)^{2}}\right]\left[6R_{2}^{2}+\Delta\left(\sqrt{16\Gamma+\left(1-\Gamma\right)^{2}}\right)-4\right]}{36R_{2}^{2}\Delta}\\
\eta_{MP\eta} & = & \frac{\left[5-\Gamma-\sqrt{16\Gamma+\left(1-\Gamma\right)^{2}}\right]\left[6R_{2}^{2}+\Delta\left(\sqrt{16\Gamma+\left(1-\Gamma\right)^{2}}\right)-4\right]}{6R_{2}^{2}\left[1+\Gamma+\sqrt{16\Gamma+\left(1-\Gamma\right)^{2}}\right]}.
\end{eqnarray}
\end{subequations} \end{widetext}

\section{Experimental verification of the theoretical energetic hierarchy
for a $\left(2\times2\right)$--electric circuit}

Just as the steady state characterized by Prigogine is associated
with the production of entropy at a minimum and constant rate in a
system, it also represents the only thermodynamic state whose useful
energy to perform work against the surroundings is zero. When we make
a distinction between spontaneous and non-spontaneous processes, we
can introduce one more constraint to the system that is related to
the extra thermodynamic conditions (operation modes). Then, all of
existing steady states bounded between the maximum power output regime
and the minimum dissipation regime (minimum entropy production), can
be physically attainable.

\begin{table}
\caption{\label{tab:Resisval}Resistor $R_{i}$ and phenomenological coefficients
$L_{jk}$ values that correspond to each of the given $q$ values.}

\centering{}\resizebox{0.5\textwidth}{!}{ %
\begin{tabular}{|c|c|c|c|c|c|c|}
\hline 
{$q$}  & {$R_{1}\,\left(\Omega\right)$}  & {$R_{2}\,\left(\Omega\right)$}  & {$R_{3}\,\left(\Omega\right)$}  & {$L_{11}\,\left(\Omega^{-1}\right)$}  & {$L_{12}\,\left(\Omega^{-1}\right)$}  & {$L_{22}\,\left(\Omega^{-1}\right)$}\tabularnewline
\hline 
\hline 
{\small{}{}0.950}  & {\small{}{}100}  & {\small{}{}1900.00}  & {\small{}{}100}  & {\small{}{}0.005128}  & {\small{}{}0.00487}  & {\small{}{}0.005128}\tabularnewline
\hline 
{\small{}{}0.955}  & {\small{}{}100}  & {\small{}{}2122.22}  & {\small{}{}100}  & {\small{}{}0.005115}  & {\small{}{}0.00488}  & {\small{}{}0.005115}\tabularnewline
\hline 
{\small{}{}0.960}  & {\small{}{}100}  & {\small{}{}2400.00}  & {\small{}{}100}  & {\small{}{}0.005102}  & {\small{}{}0.00489}  & {\small{}{}0.005102}\tabularnewline
\hline 
{\small{}{}0.965}  & {\small{}{}100}  & {\small{}{}2757.14}  & {\small{}{}100}  & {\small{}{}0.005008}  & {\small{}{}0.00491}  & {\small{}{}0.005008}\tabularnewline
\hline 
{\small{}{}0.970}  & {\small{}{}100}  & {\small{}{}3233.33}  & {\small{}{}100}  & {\small{}{}0.005007}  & {\small{}{}0.00492}  & {\small{}{}0.005007}\tabularnewline
\hline 
{\small{}{}0.975}  & {\small{}{}100}  & {\small{}{}3900.00}  & {\small{}{}100}  & {\small{}{}0.005006}  & {\small{}{}0.00493}  & {\small{}{}0.005006}\tabularnewline
\hline 
{\small{}{}0.980}  & {\small{}{}100}  & {\small{}{}4900.00}  & {\small{}{}100}  & {\small{}{}0.005005}  & {\small{}{}0.00494}  & {\small{}{}0.005005}\tabularnewline
\hline 
{\small{}{}0.985}  & {\small{}{}100}  & {\small{}{}6566.67}  & {\small{}{}100}  & {\small{}{}0.005003}  & {\small{}{}0.00496}  & {\small{}{}0.005003}\tabularnewline
\hline 
{\small{}{}0.990}  & {\small{}{}100}  & {\small{}{}9900.00}  & {\small{}{}100}  & {\small{}{}0.005002}  & {\small{}{}0.00497}  & {\small{}{}0.005002}\tabularnewline
\hline 
{\small{}{}0.995}  & {\small{}{}100}  & {\small{}{}19900.00}  & {\small{}{}100}  & {\small{}{}0.005001}  & {\small{}{}0.00498}  & {\small{}{}0.005001}\tabularnewline
\hline 
\end{tabular}} 
\end{table}

The validity of the energy conversion theorems as well as their corollaries
enunciated and developed in Section 3, is proved through the electric
model with resistors previously proposed. Taking into account the
operating conditions imposed by the well-known operating regimes ($mdf$,
$MPO$, $M\eta$, $MEF$, $M\Omega$ and $MP\eta$), we measure the
electric current (driven flux $J_{1}$) in mesh 1 of the scheme (see
Fig. \ref{fig:schemcircR}) to reproduce the hierarchical behavior
described in Section 2. That is, we are looking for the energetics
of the system for different values of the resistor $R_{2}$ and, by
transitivity for different values of $q$ that always guarantee: $X_{1}^{mdf}<X_{1}^{M\eta}<X_{1}^{MEF}<X_{1}^{M\Omega}<X_{1}^{MP\eta}<X_{1}^{MPO}$.
\begin{figure*}
\begin{centering}
\includegraphics[scale=0.73]{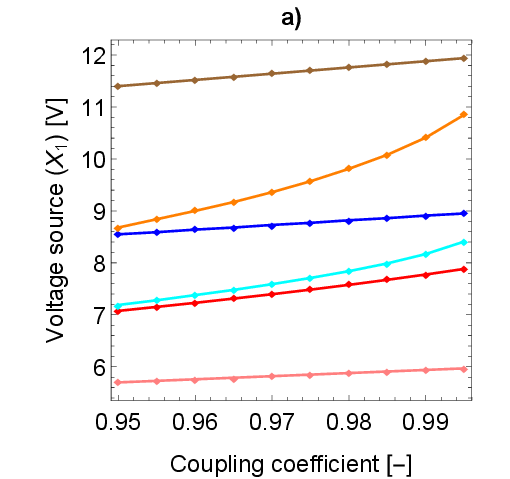}\includegraphics[scale=0.73]{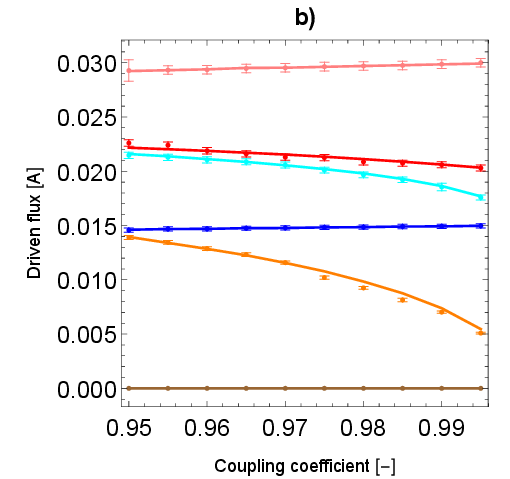}\includegraphics[scale=0.73]{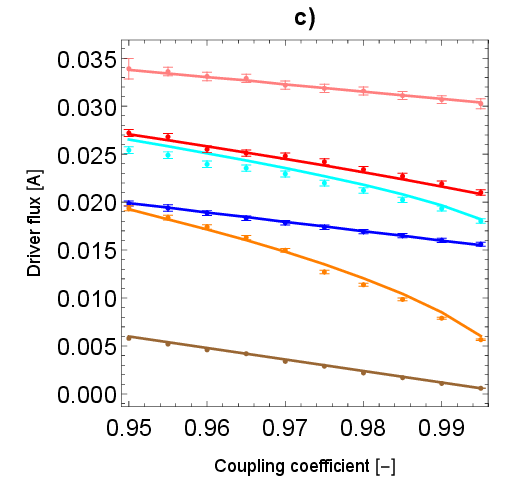}
\par\end{centering}
\caption{\label{fig:jandxvsqs}Graphs of $X_{1}$ (source voltage), $J_{1}$
(driven flux) and $J_{2}$ vs $q$. All of the points associated with
the experimental measurements are shown in diamonds while theoretical
models are depicted by solid lines. In a) a hierarchy between the
different values of $X_{1}$ is displayed. In b) and c) there is also
a trend between the operation regimes, they are expressed as the energy
consumption in each mesh. The color code for $MPO$, $MP\eta$, $M\Omega$,
$MEF$, $M\eta$ and $mdf$ are: pink, red, cyan, blue, orange and
brown, respectively.}
\end{figure*}

The nominal resistance values $R_{2}$ were calculated to check the
existence of the optimal operation regimes, by considering 10 different
values of $q\in\left[0.950,0.995\right]$, taken in steps of $0.05$.
This interval is used in analogy to the results reported by Stucki
\cite{Stucki80}, in which he considered that values of $q\in\left[0.95,0.97\right]$
show the optimal economic degrees of coupling. For this reason $R_{1}$
and $R_{3}$ were fixed at a constant value of $100\Omega$. By using
Eq. (\ref{eq:couplwitrs}), we estimated the $R_{2}$ values for each
$q$.

The next thing was to consider the triad of values for resistors $R_{1}$,
$R_{2,}$ and $R_{3}$ (see Tab. \ref{tab:Resisval}), and a fixed
value of $X_{2}=12\,\textrm{V}$ for the DC voltage source related
to the driver flux with the purpose of having completely characterized
the steady states. From Eqs. (\ref{eq:x1drivmdf}), (\ref{eq:x1drivmpo}),
(\ref{eq:x1drivmefic}), (\ref{eq:x1drivmef}), (\ref{eq:x1drivmomf})
and (\ref{eq:x1drivmep}), the values of $X_{1}$ were found as a
function of any operation mode. Thus, with the values of the resistors
$R_{i}$ and the adjusted values for the DC voltage sources $X_{k}$,
the assembled electric circuit was put inside a container with dielectric
oil to emulate its conditions as an isothermal energy converter and,
then measure the values of the driven and driver fluxes (electric
currents), by considering an enough relaxation time each $5\,\textrm{min}$
approximately to guarantee their stability (steady states).

In Fig. \ref{fig:jandxvsqs}, we display the theoretical and experimental
trend that the voltage source $X_{1}$, as well as the electron fluxes
$J_{1}$ and $J_{2}$ present, according to the operating regimes.
In the case of $X_{1}$ values, they have as upper bound $X_{2}=12V$
(the value of the fixed force). It is important to note that $X_{1}$
values in $MP$ regime, are almost halved with respect to $X_{2}$.
The purpose of the above-mentioned graphs (Fig. \ref{fig:jandxvsqs})
is to compare the behavior of $(J_{1},\,J_{2},\,X_{1})$ that the
linear energy converter model predicts with direct measurements of
electric current and voltages through a digital multimeter.

The percentage errors in the measurements for every process variable
can be estimated. In the case of power output, the maximum percentage
error is calculated for the $M\eta$ regime, its average is $\langle\Delta P_{\%}^{M\eta}\rangle\approx6.18\%$.
The efficiency presents a maximum mean percentage error when the system
works in the $M\Omega$ regime, which is $\langle\Delta\eta_{\%}^{M\Omega}\rangle\approx2.75\%$.
Finally, dissipation function have the same maximum error in the $M\Omega$
regime, whose value is $\langle\Delta\Phi_{\%}^{M\Omega}\rangle\approx7.98\%$.
In general, the foregoing shows that the fluctuations produced by
the disturbance of the system (interaction with the measuring instruments)
are small. This fact guarantees the electric circuit reaches a steady
state.

\begin{figure*}
\includegraphics[scale=0.73]{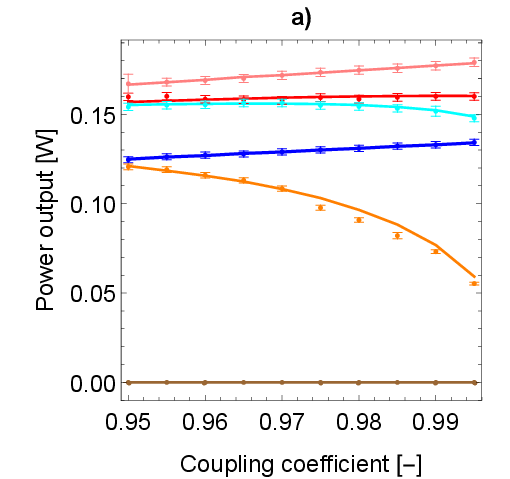}\includegraphics[scale=0.73]{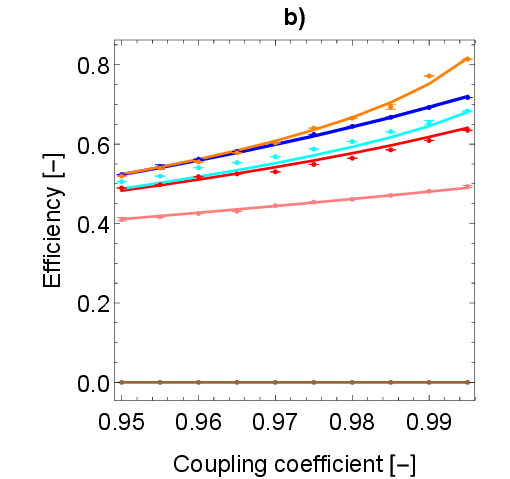}\includegraphics[scale=0.73]{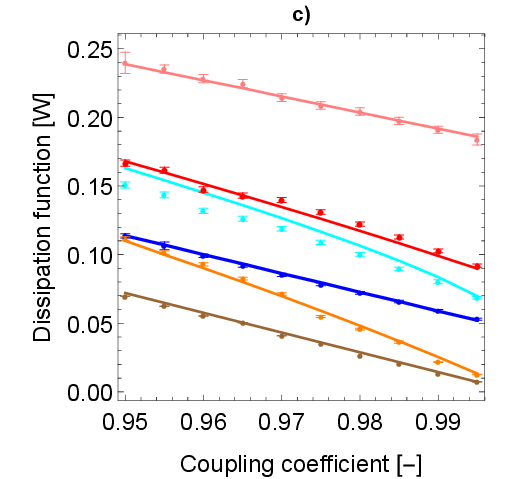} 

\caption{\label{fig:prosvarvsqs}Graphs of the three process variables vs $q$.
All of the points that correspond to the experimental measurements
are shown (in diamonds) for a) power output, b) efficiency and c)
dissipation function, while the trend given by the theoretical models
are depicted in solid lines. The color code for $MPO$, $MP\eta$,
$M\Omega$, $MEF$, $M\eta$ and $mdf$ are: pink, red, cyan, blue,
orange and brown, respectively.}
\end{figure*}

Finally, in Fig. \ref{fig:prosvarvsqs} the behavior predicted by
theorems enunciated in Section 3 is depicted for the three process
variables. Thus, all of physically attainable operation modes for
an isothermal linear energy converter are bounded between the maximum
power output and minimum dissipation function regimes, while the optimum
operation ones lies between the maximum power output and the maximum
efficiency regimes.

\section{Conclusions}

The whole set of phenomena that can be described by the so--called
steady--states, few of them have simple uncoupled processes, i.e,
a large number of phenomena have the characteristic of being spontaneous
and non-interacting (Symmetry Principle). In fact, living and man-made
systems characterized by a single process could be said to undoubtedly
satisfy Prigogine's theorem, since this extremal principle considers
unconstrained forces and therefore systems reach a diffusive regime
when energy transfer processes are carried out.

There is a divided opinion on the validity of the Minimum Entropy
Production Theorem and the Principle of Maximum Entropy Production,
which largely explain all the processes that occur near of the equilibrium
state. In our opinion, if it is desired to obtain a useful energy
available through certain coupled processes (energy conversion), and
that can be modeled to a large extent by means of so--called energy
converters within the context of LIT, then it cannot be expected a
minimum entropy production. Furthermore, we have established optimization
criteria that are associated with characteristic steady--states,
which are delimited between the $MPO$ (maximum entropy production)
and $M\eta$ operation regimes. That is, other energy conversion theorems
can be stated as long as the trade-off between the process variables
of the energy converter is well specified.

In the simple experiment that was designed, we show on the one hand,
that the way in which voltage sources (thermodynamic forces) are tuned,
lead us to establish a unique performance and on the other hand, that
the steady--states associated with an energetic objective function
are physically accessible. Furthermore, the state described by the
minimum dissipation function leads us to a zero energy conversion.
\begin{acknowledgments}
This work was partially supported by Instituto Polit\'ecnico Nacional:
SIP--project numbers: 20221365 and 20221415, COFAA--Grants: 5406 and 5419, EDI--Grants: 1750 and 2248; and SNI--CONACYT Grants: 10743
and 16051, MEXICO.
\end{acknowledgments}

\vspace*{0.5cm}

\appendix
%dummy

\section{Some algebraic properties of the optimal performance regimes}

From entropy production definition as a function of the spontaneous
and non--spontaneous forces, the surface $\Phi$ characterized by
an ordered pair $\left(X_{1},X_{2}\right)$, lead us to define a vector
space given by the basis vectors $\mathcal{X}=\{(X_{1},0),(0,X_{2})\}$.
Besides, as it is possible to express $X_{1}$ as a function of $X_{2}$,
then 
\begin{eqnarray*}
H_{(X1,X2)} & := & \{(X_{1},X_{2})|X_{1}=AX_{2},\text{ with }\\
 &  & A\text{ a fixed and }X_{2}\text{ an arbitrary }\in\Re\}
\end{eqnarray*}
defines the proper subspace of such vector space \cite{Grossman94}.
An energy converter delimits the vector space of these linear irreversible
processes in the region with $X_{1}<0$ and $X_{2}>0$. From Eqs.
(\ref{eq:opXvp}) and (\ref{eq:opXvp2}), we note that these optimal
operation modes lie in the subspace $H_{(X1,X2)}$. That is, according
to the physical information of $A$, a linear energy converter can
access different physical realizations as long as the thermodynamic
forces are tuned under the respective constraints.
\begin{figure}
\begin{centering}
\includegraphics[scale=0.53]{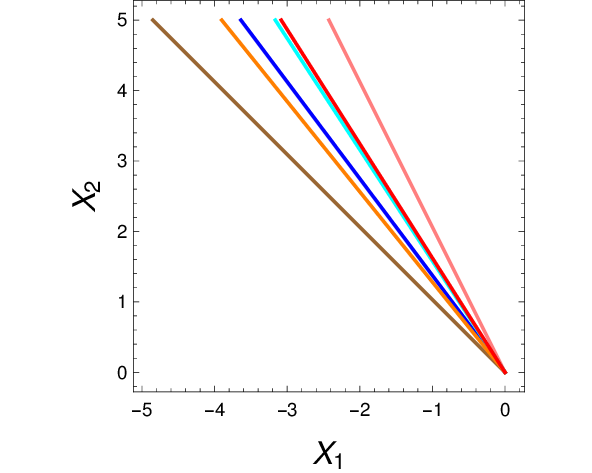}
\par\end{centering}
\caption{\label{fig:spacX1X2}Graphic sketch of the proper sub-space $H_{(X1,X2)}$
the for different values of $A'$, $X_{1}$ and $X_{2}$ are the so-called
generalized thermodynamic forces. An order is noted for the 6 operation
regimes, \textit{mdf} in blue, \textit{M$\eta$} in cyan, \textit{MEF}
in green, \textit{M$\Omega$} in black, \textit{MP$\eta$} in purple
and \textit{MPO} in red. All lines were depicted under the normalization
condition $A'=\nicefrac{A\sqrt{L_{11}}}{\sqrt{L_{22}}}$ and with
$q=0.97$.}
\end{figure}

Since the restriction $X_{1}^{Y}=A^{Y}X_{2}$, with $Y$ the optimal
operation mode, is associated with a particular extra-thermodynamic
condition; in Fig. \ref{fig:spacX1X2} we can observe the geometric
representation of proper sub-spaces of $\mathcal{X}$, restricted
to the physical region of linear energy converters given by the Eqs.
(\ref{eq:opXvp}) and (\ref{eq:opXvp2}).

The order displayed in Fig. \ref{fig:spacX1X2} for each operating
regime, can be viewed as the availability of a linear energy converter
to locate physical realizations when thermodynamic forces are associated
in such a way as to achieve a particular goal in energy conversion.
That is, if we measure the distance between an arbitrary point $(X_{1},X_{2})$
and the equilibrium state $(0,0)$ for each proper subspace, 
\[
d(\vec{0},\vec{X})=||\vec{X}-\vec{0}||=X_{2}\sqrt{1+(A^{Y})^{2}},
\]
we observe that $A^{MPO}<A^{MP\eta}<A^{M\Omega}<A^{MEF}<A^{M\eta}<A^{mdf}$
and therefore the distances revels the same order as inequalities
{[}see Eq. (\ref{eq:hierDPE}){]}.

\end{document}